\documentclass[runningheads, a4paper]{llncs}

\usepackage[pdftex]{graphicx} 

\usepackage{amsfonts,amssymb,amsbsy}
\usepackage{fullpage}

\usepackage{float}
\usepackage{algorithm2e}

\title{The min-max edge $q$-coloring problem}
\author{Tommi Larjomaa \and Alexandru Popa}
\institute{
Department of Communications \& Networking, Aalto University School of Electrical Engineering, Aalto, Finland. \\ E-mails: \email{\{tommi.larjomaa@aalto.fi, alexandru.popa@aalto.fi\}}
}

\begin{document}
\maketitle
\begin{abstract}
In this paper we introduce and study a new problem named \emph{min-max edge $q$-coloring} which is motivated by applications in wireless mesh networks. The input of the problem consists of an undirected graph and an integer $q$. The goal is to color the edges of the graph with as many colors as possible such that: (a) any vertex is incident to at most $q$ different colors, and (b) the maximum size of a color group (i.e. set of edges identically colored) is minimized. We show the following results:
\begin{enumerate}
\item Min-max edge $q$-coloring is NP-hard, for any $q \ge 2$.
\item A polynomial time exact algorithm for min-max edge $q$-coloring on trees.
\item Exact formulas of the optimal solution for cliques and almost tight bounds for bicliques and hypergraphs. 
\item A non-trivial lower bound of the optimal solution with respect to the average degree of the graph.
\item An approximation algorithm for planar graphs.
\end{enumerate}
\end{abstract}

\section{Introduction}


Traditionally, backbone connectivity in networks of various sizes has been built using wired infrastructure. Even though the bandwidth that modern wired networking technology offers is no doubt better than that of wireless alternatives, the material and installation costs of wired networks is a significant drawback. Therefore, the concept of wireless mesh networks (WMNs) has received a lot of attention and has been researched actively during the past decade~\cite{Akyildiz2005445,Wang2008,draves2004,gupta2009,koksal2006}.

In a multi-channel WMN, each node is able to use multiple non-overlapping frequency channels. The use of many channels inside the same network can significantly improve overall performance; interference from neighboring nodes can be decreased substantially, when nodes do not need to use the same radio channel for every link. Multiple radio channels in the network means, that at least some of the nodes need to handle more than one channel at a time. In many proposed designs the multi-channel feature is achieved by packet-by-packet reconfiguration of the radio \cite{muir1998,kyasanur2005,so2004}. However, one of the drawbacks of this kind of continuous channel switching of a single radio interface is that it requires precise synchronization throughout the network.

An alternative approach would be to fit multiple radio interfaces to each node, thus allowing a more persistent channel allocation per interface. A couple of such multi-NIC (network interface card) architectures have been proposed by Raniwala et al. \cite{RaniwalaGC04,RaniwalaC05}. Their simulation and testbed experiments show a promising improvement with only two NICs per node, compared to a single-channel WMN. Another appealing feature of these architectures is that they are based on readily available, commodity IEEE 802.11 interfaces, requiring only systems software modification. 

The scenario of two or more NICs per node with fixed channels imposes some limitations to the assignment of channels on each interface. In order to set up a link between two nodes, both of them have to have at least one of their interfaces set to the same channel. On the other hand, links inside an interference range should use as many different channels as possible. Thus, the channels need to be assigned carefully in order to both keep every required link possible and maximize useful bandwidth throughout the network.

The channel assignment problem can be modelled as a type of edge coloring problem: given a graph $G$, the edges have to be colored so that there are at most $q$ different colors incident to each vertex. Here, vertices, edges and colors represent network nodes, links and channels, respectively. A coloring that satisfies this constraint, is called an edge $q$-coloring. Note, that the coloring constraint differs from the traditional coloring problems, where adjacent items are not allowed to have the same color. Also the goal is different; instead of minimizing the amount of colors, a large amount of different colors in an edge $q$-coloring is often a desired state of things.

Previously, the channel assignment was formulated as the max edge $q$-coloring problem, where the goal was to maximize the total number of edges. The drawback of this model is that in an optimal solution the same color is assigned to many edges while other colors are used only once. Thus, max edge $q$-coloring does not reflect the problem motivation and it is more realistic to try to have the color components as balanced as possible. Therefore, we newly introduce the min-max edge $q$-coloring where the goal is to minimize the maximum size of a color group.

\textbf{Related work.} The problem of finding a maximum edge $q$-coloring of a given graph has been first studied by Feng et al.~\cite{FengZQW07,FengCZ08,FengZW09}. They provide a $2$-approximation algorithm for $q=2$ and a $(1+\frac{4q-2}{3q^2-5q+2})$-approximation for $q > 2$. They show that the problem is solvable in polynomial time for trees and complete graphs in the case $q=2$, but the complexity for general graphs has been left as an open problem. Later, Adamaszek and Popa~\cite{AP10} show that the problem is APX-hard and present a $5/3$-approximation algorithm for graphs which have a perfect matching.
The maximum edge $q$-coloring is also considered in combinatorics and is a particular case of the anti-Ramsey number. For a brief description of the connection of the two problems, the reader can refer to~\cite{AP10}.

To the best of our knowledge there has been no prior research on the min-max edge $q$-coloring problem.


\textbf{Our contributions.}
In this paper we introduce and study the min-max edge $q$-coloring problem. First,  in Section~\ref{sect:mmenphard} we prove that the problem is NP-hard for any $q \ge 2$. The proof is split into two parts. We first show the NP-hardness for a more general version in which each vertex is allowed to have an independent value of $q$. In the second part we show how to introduce extra gadgets in order to force all the vertices to have the same value of $q$. 

Then, in Section~\ref{sect:polytree} we show an exact polynomial time algorithm for trees, for $q=2$. We first show that the optimal solution in a tree is at least $\Delta / 2$ and at most $\Delta$, where $\Delta$ is the maximum degree of the tree. Then, the algorithm uses binary search to find a value $c$, such that the input admits a coloring in which the largest color group is at most $c$. Given a value $c$, we select in turn each vertex as the root of the tree and try to construct a solution in a bottom up fashion. This is not straightforward as for each vertex we have to solve a knapsack instance (fortunately, these instances are solvable in polynomial time).

In Section~\ref{sect:special} we analyse the value of the optimal solution on special classes of graphs: cliques, bicliques and hypercubes. We provide the exact formulas of the optimal solutions for cliques. For bicliques we present a lower bound which is tight when both parts of the graph have an even number of vertices (and almost tight for the other cases). For a hypergraph $Q_n$ we give a lower bound which is tight for even $n$, and similarly, almost tight for odd $n$. Although these classes of graphs have a very simple structure, finding lower bounds is much more difficult than in the case of the max edge $q$-coloring problem. We need to prove several lemmas in order to understand the structure of the coloring and to prove the final theorems.

A good lower bound of the optimal solution is necessary in order to design approximation algorithms. For the min-max edge $q$-coloring problem, a trivial lower bound is half of the maximum degree. Nevertheless, in Section~\ref{sect:lb-ub} we show another lower bound in terms of the average degree of the graph. Section~\ref{sect:separator}  presents an approximation algorithm for planar graphs which achieves a sublinear approximation ratio. The algorithm uses a theorem of Lipton and Tarjan~\cite{Lipton1980} which says that a planar graph admits a small balanced separator.
Section~\ref{sect:summary} summarizes the results and briefly discusses possible future research directions.


\section{NP-hardness of Min-max Edge $q$-coloring}\label{sect:mmenphard}

In this section we prove that the min-max edge $q$-coloring problem is NP-hard for $q \geq 2$, giving little hope of finding a general exact polynomial time algorithm for it. The proof is split into two steps. First we prove NP-hardness for a more general version of the problem, defined next, where each vertex is assigned a value for $q$ individually.

\begin{problem}[General min-max edge $q$-coloring problem]\label{problem:general-mme}

The input is a graph $G = (V,E)$, and for each vertex $v_i$ there is a positive integer $q_i$. A feasible solution is a coloring of edges, such that for each vertex $v_i$, there are at most $q_i$ different colors incident to it. The goal is to find a coloring $\sigma$ such that the size of the largest color group, $\displaystyle\max_c |\{ e \in E | \sigma(e) = c \}|$, is minimized.
\end{problem}

The reduction is made from monotone one-in-three SAT (Definition \ref{problem:one-in-three-sat}), which is known to be NP-complete \cite{schaefer1978}. By modifying this reduction slightly we can prove NP-hardness for the min-max edge $q$-coloring problem with a constant value of $q$.

\begin{definition}[monotone one-in-three SAT problem]\label{problem:one-in-three-sat}

The input is a Boolean 3CNF-formula $\phi$, where each literal is simply a variable; there is no negation. Determine whether a truth assignment for the variables exists, such that for each clause, there is exactly one literal that is true, while the other two literals are false.
\end{definition}

Now we state and prove NP-hardness for the general edge $q$-coloring problem.

\begin{theorem}\label{theorem:general-mme-np-hard}
Problem \ref{problem:general-mme} is NP-hard.
\end{theorem}
\begin{proof}
We use a reduction from monotone one-in-three SAT (Definition \ref{problem:one-in-three-sat}), which goes as follows. There are $m$ clauses and $n$ variables in the formula $\phi$. For each clause, there is a single vertex $c_j$ with $q=2$ (we use this notation as a shorthand for ``at most 2 different colors can be incident to $c_j$''). For each variable $x_i$, there are three vertices: $a_i$, $b_i$ and $v_i$ having $q=1$, $q=1$ and $q=2$, respectively. Each vertex $v_i$ is adjacent to vertices $a_i$ and $b_i$. If a variable is present in a clause, the corresponding variable vertex $a_i$ is adjacent to the clause vertex $c_j$. For each $a_i$, there are additional leaves adjacent to it, so that $\deg(a_i) = 2m_i$, where $m_i$ is the number clauses the variable is present in.\footnote{We can safely assume that each variable has at most one literal in any clause. Otherwise a variable set to true could make a clause unsatisfied regardless of other variables, essentially making the formula not 3CNF from the point of view of the satisfiability problem.} Moreover, for each $b_i$, there are additional leaves so that $\deg(b_i) = L - 2m_i$, where $L = 4m + n$.  Finally, there is a vertex $f$ with $q=1$, that is adjacent to each $v_i$. The resulting graph is of polynomial size in $m$. Figure \ref{img:1-in-3sat-minmax-reduction} illustrates the idea of the reduction by showing the full gadget of a single variable.

\begin{figure}[htb]
\centering \includegraphics[scale=1]{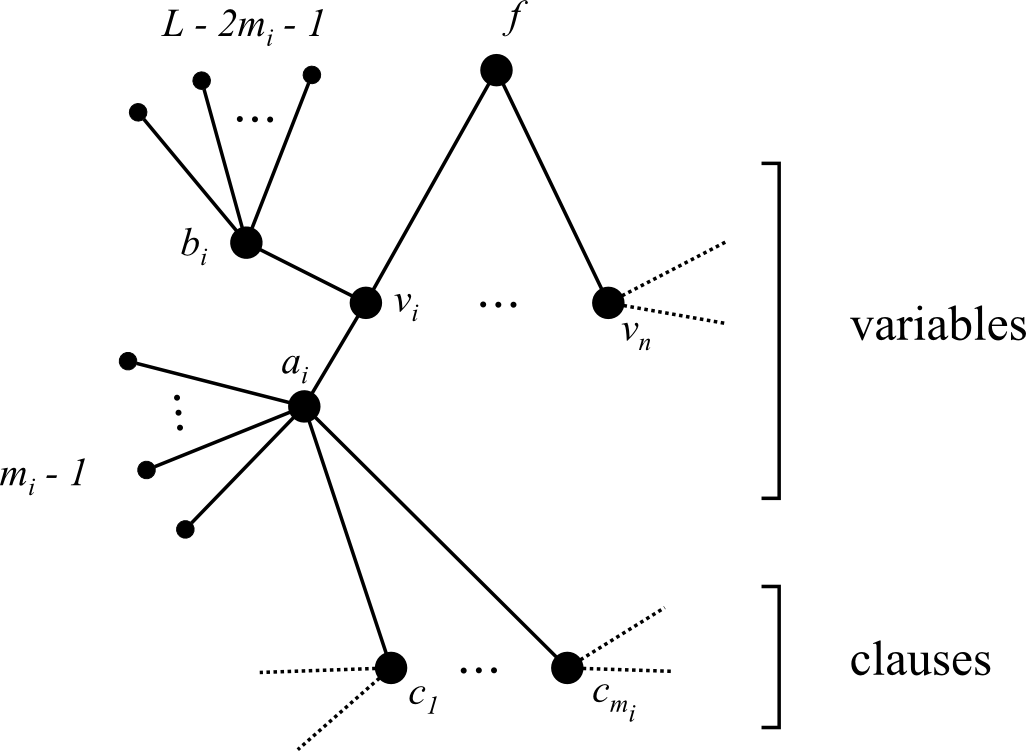}
\caption{The reduction from monotone 1-in-3 SAT to general min-max edge $q$-coloring for a single variable. The dotted edges lead to similar variable gadgets.}\label{img:1-in-3sat-minmax-reduction}
\end{figure}

Next we show that if $\phi$ is satisfiable, there is a feasible coloring for the reduction, whose largest color group is $L$. For each variable $x_i$ that is false in the satisfying truth assignment, color the edges incident to $a_i$ with the color $F$, which is the color incident to the vertex $f$. There are two edges incident to some $a$-vertex per each literal, and $2m$ false literals, so there are in total $4m + n = L$ edges colored with $F$. Since $v_i$ is incident to only one color at this point, we give a distinct color for the edges of $b_i$, of which there are less than $L$.

For each true variable $x_i$ we choose a distinct color and use it to color edges incident to both $a_i$ and $b_i$. These color groups have thus $L - 2m_i + 2m_i = L$ edges. Since the truth assignment is satisfying, there is one color representing a true variable and the color $F$ representing false variables incident to each clause vertex, which makes the coloring feasible.

Finally, we show that if the formula is not satisfiable, the optimum of the reduction is more than $L$ (in other words, if the optimum of the reduction is less or equal to $L$, the formula is satisfiable). In a feasible coloring of a reduction from an unsatisfiable formula, there are two possibilities. Either there are clauses in which two or more variables and their $a$-vertices have a color different from $F$, or there are clauses in which all variables are using color $F$ (or both).

In the first case, two variable vertices $a_i$ and $a_j$ necessarily share a color, which we denote by $C$. Consequently, the vertices $v_i$ and $v_j$ are both saturated with colors $F$ and $C$. Note, that for any variable $x_k$, $\deg(b_k) \geq L - 2m > L/2$. Thus, if the vertices $b_i$ and $b_j$ are assigned the same color, the limit $L$ is immediately exceeded. On the other hand, if one of those vertices, say, $b_i$ takes the color $F$, and $b_j$ takes the color $C$, there are already $L$ edges colored with $C$ due to the variable $x_j$ plus the edges incident to $a_i$.

In the second case we can assume that the clauses that have not only false literals in them, have exactly one true literal, since the other case was already discussed. Now, there are more than $2m$ false literals, and, as observed before, there are two edges per literal incident to the $a$-vertices. Thus, there must be more than $4m + n = L$ edges colored with $F$. \qed
\end{proof}

As we go on to prove NP-hardness for min-max edge $q$-coloring, where each vertex has the same value for $q$, we use a slightly modified version of the previous reduction. The idea is to mimic vertices with $q=1$ or $q=2$. This is done by saturating vertices with an appropriate number of different colors that already have $L$ edges. We proceed with the theorem and proof.

\begin{theorem}
The min-max edge $q$-coloring problem is NP-hard for $q \geq 2$.
\end{theorem}
\begin{proof}
We begin by showing how to force a vertex with any value of $q$ to allow only one or two new colors for its additional edges, given the upper bound $L$ for color group size. Observe that the optimum for a $(qL+1)$-star, namely a star with $qL$ leaves, is exactly $L$. We take $q-1$ such stars, pick one leaf from each star and contract them as one vertex. In an optimal coloring of the acquired gadget, the contracted vertex $v$ is incident to $q-1$ different colors of size $L$. As we add edges to $v$, they can be colored with only one color in order to keep color group sizes below $L$. If we want a vertex that allows two colors, we pick $q-2$ leaves from different $qL$-stars (we can use the same stars as before, since there are plenty of leaves left) and contract them as one.

Using such gadgets that mimic vertices with $q=1$ and $q=2$, we straightforwardly construct a reduction equivalent to the one used in the proof of Theorem \ref{theorem:general-mme-np-hard}. Now it remains to show that the number of additional vertices and edges in the new reduction is polynomially bounded in the size of the formula.

We show that we need only ($q-1$) stars to be able to mimic enough vertices. In the original reduction, there is one vertex per clause, three vertices per variable\footnote{We do not need to take into account the leaves of the variable vertices; a leaf allows only one color incident to it, no matter what value $q$ has.} and the vertex $f$. In total we have $M = m + 3n + 1$ vertices that need to be mimicked. We need at most $q-1$ leaves to mimic one vertex, so having $qL \geq 2L \geq M$ will suffice. Assume the opposite, which yields
\begin{displaymath}
M > 2L \Leftrightarrow m + 3n + 1 > 8m + 2n \Leftrightarrow n > 7m - 1.
\end{displaymath}
This contradicts with the fact that there can be at most $3m$ variables in a 3CNF-formula, that is, $n \leq 3m$. So, the number of additional edges needed for the modified reduction is $(q-1)qL = O(m+n)$, since $q$ is constant. \qed
\end{proof}

\section{Exact Polynomial Time Algorithm for Trees}\label{sect:polytree}

In this section we present an exact polynomial time algorithm for solving the min-max edge 2-coloring problem on trees. First of all we give the following bound of the optimal solution.

\begin{lemma}\label{lemma:minmax-opt-trees}
For an instance of the min-max edge 2-coloring problem, where the graph is a tree $T$, OPT $\in \left[ \frac{\Delta}{2}, \Delta-1 \right]$, where $\Delta$ is the maximum degree of $T$.
\end{lemma}
\begin{proof}
The lower bound follows from the fact that there is a vertex with $\Delta$ edges incident to it, and only two distinct colors can be assigned to these edges. The upper bound can always be achieved with the following coloring. Choose an arbitrary vertex $v_r$ as the root vertex, and color its edges evenly with two colors. For each child $v$ of $v_r$, there are $deg(v)-1$ uncolored edges that can be colored with a new color, since $v$ had only one edge colored previously. The same is repeated iteratively for each child vertex of a visited vertex. No more than $\Delta-1$ edges are colored with any color. \qed
\end{proof}

The polynomial time algorithm for trees is defined below (Algorithm \ref{alg:polytree}). The idea of the algorithm is to try to color the tree with different candidate values for optimum from the interval $\left[ \frac{\Delta}{2}, \Delta-1 \right]$, until candidates $c$ and $c-1$ are found so that $c$ leads to a feasible coloring whereas $c-1$ does not. This is repeated for each vertex as the root vertex, and the smallest successful value of $c$ is the optimum. By applying the principle of binary search we only need to test $O(\log \Delta)$ different candidates per root. Now we prove that the output is in fact optimal.

\begin{theorem}
Given a tree as input, the output of Algorithm \ref{alg:polytree} is a feasible and optimal solution to the min-max edge 2-coloring problem.
\end{theorem}
\begin{proof}
From Lemma \ref{lemma:minmax-opt-trees} we know that the optimum is within the search range of the algorithm, so if it is able to identify a feasible maximum color group size candidate, it finds the optimum. To see that this is the case, we analyse the applied coloring strategy carefully.

As in the algorithm, we choose one vertex of the input tree $T$ at a time as root. Consider a maximum candidate $c$ and a non-leaf vertex $v$ that has only leaves as children. In order to keep the color group sizes below $c$, it is best to color as many leaf edges of $v$ as possible with one color. Anything less would be more detrimental for the task of satisfying the limit $c$, since the parent of $v$, namely $v_p$, needs to use one of its two colors for the residual edges of $v$. In other words, when the edge between $v$ and $v_p$ is assigned a color, the color necessarily propagates to the child edges of $v$ that are yet without a color. Note also that the parent edge of a vertex is necessarily one of its residual edges. Thus, for any non-root vertex the residual number (introduced in step 6 of the algorithm) is at least one, whereas for the root it can be zero.

In addition to $v$, its parent $v_p$ possibly has other children that similarly to $v$ have a certain amount of uncolored residual edges. This is where the knapsack problem comes in. One color needs to be assigned to a set of children of $v_p$ so that the sum of the residual numbers of these children is maximized, but does not exceed $c$. This in turn minimizes the residual number of $v_p$. We then repeat this minimization task for each vertex. This needs to be done in a bottom up order since the residual number of a vertex is dictated by those of its children. If at any point the residual number of a vertex turns out larger than $c$, we know that with the currently chosen root vertex, the coloring attempt fails to satisfy the candidate limit. If all residual numbers are less than $c$, the coloring is successful. Figure \ref{img:tree-fail} illustrates a failed coloring attempt.


\begin{algorithm}[H]
\caption{Tree 2-coloring algorithm}
\label{alg:polytree}
\textbf{Input:} A tree graph $T$
\begin{tabbing}
1. $m \longleftarrow \Delta - 1$ \\
2. \= For each vertex $v^r$ \\
   \> 3. Label each vertex of $T$ with its distance from the root vertex (via e.g. BFS) \\
   \> 4. $l \longleftarrow \left\lceil \frac{\Delta}{2} \right\rceil, u \longleftarrow \Delta - 1$ \\
   \> 5. \= Repeat \\
   \>    \> 6. \= Assign for each non-root vertex $v$ a residual number $v_l \longleftarrow 1$, and for \\
   \>    \>    \> the root $v^r_l \longleftarrow 0$ \\
   \>    \> 7. $c \longleftarrow \left\lfloor \frac{1}{2}(l+u) \right\rfloor$ \\
   \>    \> 8. \= For each non-leaf vertex in descending order of distance from root \\
   \>    \>    \> 9. \= Solve the following knapsack instance: \\
   \>    \>    \>    \> Denote the children of $v$ by $v^i$. Size of the knapsack is $c$, and the item \\
   \>    \>    \>    \> sizes are the residual numbers $v_l^i$ of the children. \\
   \>    \>    \> 10. Store the set of indices of the children in the knapsack solution to $S$ \\
   \>    \>    \> 11. If $\displaystyle\sum_{i} v_l^i - \displaystyle\sum_{j \in S} v_l^j + v_l > c$: $l \longleftarrow c+1$ and go to step 16 \\
   \>    \>    \> 12. \= Color the uncolored edges incident to $v^i, i \in S,$ and all their successors \\
   \>    \>    \>    \> with a new color \\
   \>    \>    \> 13. $v_l \longleftarrow v_l + \displaystyle\sum_{i} v_l^i - \displaystyle\sum_{j \in S} v_l^j$ \\
   \>    \> 14. Color the remaining uncolored edges connected to the root with one color \\
   \>    \> 15. Store the current coloring to $U$, and set $u \longleftarrow c$ \\
   \>    \> 16. If $l=u$, revert to the coloring $U$, jump out of the loop to step 17 \\
   \> 17. if $u < m$: $m \longleftarrow u$ and $M \longleftarrow U$ \\
18. Revert to the coloring $M$ \\
\end{tabbing}
\textbf{Output:} $m$
\end{algorithm}

Essentially, one run through the loop starting at step 8 minimizes the residual number of the root vertex with respect to an optimum candidate $c$. If a residual number exceeds $c$, the combination of the root vertex and the optimum candidate does not lead to a feasible coloring. Changing the root vertex, however, changes the parental relationships between the vertices, and consequently the residual numbers, even if the optimum candidate was the same. This is why we need to iterate the minimization process with all combinations of root vertices and optimum candidates to be sure\footnote{It might be that a failure to color a tree with any root $v^r$ and a fixed optimum candidate $c$ implies a similar failure for all possible roots, but this remains an open question.}. Since there are merely $O(n\Delta)$ of such combinations, this does not compromise the algorithm running in polynomial time.

\begin{figure}[htb]
\centering \includegraphics[scale=1]{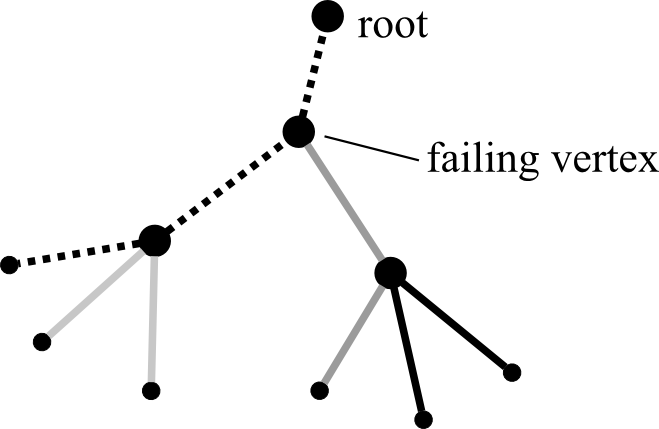}
\caption{A failed attempt to color a tree with optimum candidate 2. The dashed edges are without color.}\label{img:tree-fail}
\end{figure}

As a final note, since the knapsack problem is known to be NP-hard, it might give reason to believe that step 9 of the algorithm does not run in polynomial time in general. Fortunately, it is also well known that knapsack instances are solvable in $O(nW)$ time, where $n$ is the number of items and $W$ is the size of the knapsack. Since at any vertex there are at most $\Delta$ items (children) and the knapsack size is also at most $\Delta$, any knapsack instance encountered in step 9 is solvable in $O(\Delta^2)$ time. \qed

\end{proof}

\section{Special Cases}\label{sect:special}

In this section we present formulas for the optimal solution of the min-max edge 2-coloring problem in the case of three simple graph types. These special cases are cliques, bicliques and hypercubes.

\subsection{Clique}\label{ssect:clique}

Here we show that an optimal min-max edge 2-coloring of an $n$-clique $K_n$ achieves OPT$(K_n) \geq \left\lceil \frac{1}{3}|E(K_n)| \right\rceil$. The proof is split in parts. Also, we show that the bound is tight in most cases and present exact formulas for the optimum in all cases. Before we begin, we define a term used frequently later on.


\begin{definition}[Color subgraph]\label{def:color-subgraph}
For a given feasible edge $q$-coloring of a graph $G$ and a color $c$, a \textit{color subgraph} $G^{c}$ is an edge induced subgraph of $G$, induced by all edges with the color $c$.
\end{definition}

The first observation concerns a color that is not incident to every vertex of the clique. Such a color can share vertices with only a limited number of other colors. This and the forthcoming lemmas help narrow down the different ways of how a clique can be colored. 

\begin{lemma}\label{lemma:clique1}
In a feasible edge 2-coloring of a clique $K_n$ and for any color $c$, a color subgraph $K_n^{c}$ cannot share vertices with more than two other color subgraphs, if $V(K_n^{c}) \subset V(K_n)$.
\end{lemma}
\begin{proof}
Assume the opposite. In a feasible coloring of $K_n$, let $K_n^{c}$ be a color subgraph that shares vertices with $k \geq 3$ other color subgraphs $K_n^{c_1}, \ldots, K_n^{c_k}$, and $V(K_n^{c}) \subset V(K_n)$. Now, a vertex $v$ in $V(K_n^c)$ is incident to two colors: $c$ and $c_i$, the latter being assigned to the edges going from $v$ to vertices not in $V(K_n^c)$.  Formally, $V(K_n) \setminus V(K_n^{c}) \subset V(K_n^{c_i})$ for each $i = 1, \ldots, k$. Thus, we have a set of vertices $V(K_n^{c_i})$ that is incident to $k$ colors, which makes the coloring not feasible, a contradiction. \qed
\end{proof}

Next, we look at a more specific case of the situation described in the above lemma. When a color is not incident to all vertices and shares vertices with exactly two other colors, there are exactly three colors, all of which are necessarily incident to the other two. This coloring strategy actually turns out to be the best in the end.

\begin{lemma}\label{lemma:clique2}
Given a feasible edge 2-coloring of $K_n$, for which there is a color subgraph $K_n^{c}$ that shares vertices with exactly two other color subgraphs, and $V(K_n^{c}) \subset V(K_n)$, the coloring has exactly three colors, whose color subgraphs have these same properties.
\end{lemma}
\begin{proof}
Let $K_n^{c}$ be a color subgraph of $K_n$ that shares vertices with exactly two other color subgraphs $K_n^{c_1}$ and $K_n^{c_2}$. As in the proof of Lemma \ref{lemma:clique1}, $V(K_n) \setminus V(K_n^{c}) \subset V(K_n^{c_i}), i = 1,2$. Thus, all vertices $V(K_n) \setminus V(K_n^{c})$ are saturated with 2 colors. Since $V(K_n^{c})$ was assumed to be incident to only the three colors, there cannot be any other colors. Furthermore, none of the colors is incident to all vertices. \qed
\end{proof}

In the following lemma we cover the remaining non-trivial alternative, that is, there is a color that shares vertices with exactly one other color. This implies the presence of a color incident to all vertices. From now on we call such a color \textit{global}.

\begin{lemma}\label{lemma:clique3}
Given a feasible edge 2-coloring of $K_n$ and a color subgraph $K_n^c$ that shares vertices with exactly one other color subgraph $K_n^F$, and $V(K_n^c) \subset V(K_n)$, the color $F$ is incident to all vertices of $K_n$.
\end{lemma}
\begin{proof}
The edges between $V(K_n^c)$ and the rest of the vertices must be colored with some other color than $c$. Since $c$ is incident only to $V(K_n^c)$, the edges between $V(K_n^c)$ and the rest of the vertices must be colored with $F$. Thus, $F$ is incident to all vertices of $K_n$. \qed
\end{proof}

We now have enough tools to provide the actual lower bound. First we show that if there are more than four colors, one of them must be global. This, in turn, yields that one of the colors has over one third of all edges. Since the alternative is to have three or less different colors, the lower bound follows.

\begin{theorem}\label{theorem:clique}
For min-max edge 2-coloring, the following holds:
\begin{equation}\label{eqn:clique-lb}
\textrm{OPT}(K_n) \geq \left\lceil \frac{1}{3} |E(K_n)| \right\rceil = \left\lceil \frac{n(n-1)}{6} \right\rceil
\end{equation}
\end{theorem}
\begin{proof}
First of all, we observe that in order to have $OPT(K_n) < \left\lceil \frac{1}{3} |E(K_n)| \right\rceil$, at least four different colors must be used in an optimal coloring. Assume this is possible. With at least four colors, Lemmas \ref{lemma:clique1} and \ref{lemma:clique2} imply that the colors not incident to all vertices can share vertices with only one other color. By Lemma \ref{lemma:clique3}, that other color is the global color $F$. Now, let $K_n^{c}$ be the color subgraph with the largest proper subset of vertices of $K_n$, and let $k_{c} = |V(K_n^{c})|$. Edges of only the global color fill the cut (i.e. the set of edges between two groups of vertices) between $V(K_n^{c})$ and the rest of the $n-k_{c}$ vertices, thus
\begin{displaymath}
k_{c}(n-k_{c}) \leq \frac{1}{3} |E(K_n)| = \frac{n(n-1)}{6}.
\end{displaymath}
With the help of basic calculus, this yields
\begin{equation}\label{eqn:clique1}
k_{c} \leq \frac{n}{2} - \sqrt{\frac{n^2 + 2n}{12}} < \left(\frac{1}{2} - \frac{1}{\sqrt{12}}\right)n < \frac{1}{3}n
\end{equation}
or
\begin{equation}\label{eqn:clique2}
k_{c} \geq \frac{n}{2} + \sqrt{\frac{n^2 + 2n}{12}} > \left(\frac{1}{2} + \frac{1}{\sqrt{12}}\right)n > \frac{2}{3}n.
\end{equation}
If (\ref{eqn:clique1}) is true, there are two possibilities: either all non-global colors are incident to a total of less than a third of all vertices, or there is a set of non-global colors that are incident to a total of $k$ vertices, so that $\frac{1}{3}n \leq k \leq \frac{2}{3}n$. In the former case, $|E(K_n^F)| > \frac{1}{3} |E(K_n)|$, a contradiction. In the latter case, the cut between the $k$ and the other $n-k$ vertices are again filled with edges of the global color, as in the case of $k_{c}$, but this time $k$ fails to satisfy (\ref{eqn:clique1}) or (\ref{eqn:clique2}), leading to a contradiction.

If (\ref{eqn:clique2}) is true, there are $k < \frac{1}{3}n$ vertices for the rest of the colors to occupy. In total, these vertices have at most the following amount of edges between them:
\begin{displaymath}
|E(K_k)| = \frac{k(k-1)}{2} < \frac{\frac{1}{3}n^2 - n}{6} < \frac{n^2-n}{6} = \frac{1}{3}|E(K_n)|.
\end{displaymath}
Thus, over two thirds of the edges are left for the two other colors to share, leaving the lower bound out of reach.

Now, the only way to achieve the suggested lower bound is by using three colors, in which case the bound is trivial. \qed
\end{proof}

Most of the time, the lower bound is actually tight, and it is achievable only with a coloring described in Lemma \ref{lemma:clique2} (i.e. every vertex incident to exactly two colors, no global color) for two reasons. First, as we saw in the above proof, the lower bound is out of reach using four colors. Second, if one of the three colors is global, the other colors need to satisfy either (\ref{eqn:clique1}) or (\ref{eqn:clique2}), leaving at least one of them too small. Figure \ref{img:clique-coloring} shows an optimal coloring of $K_6$.

\begin{figure}[htb]
\centering \includegraphics[scale=0.75]{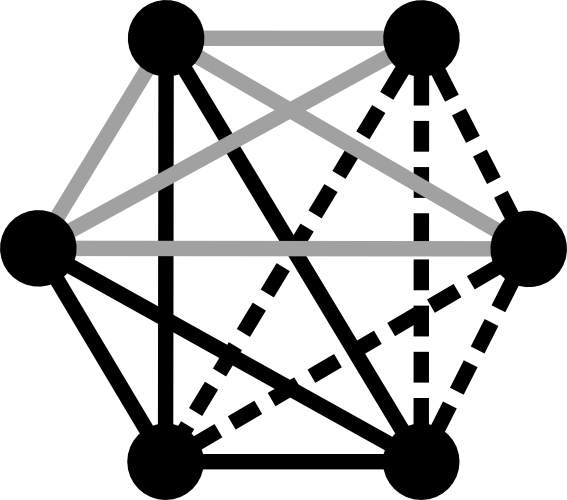}
\caption{An optimal coloring of $K_6$.}\label{img:clique-coloring}
\end{figure}

The most even way to distribute the edges between three colors is to first divide the vertices of $K_n$ to three groups of sizes $k = \left\lfloor \frac{n}{3} \right\rfloor$ and $k + 1$, depending on the remainder of the division. Each color is then incident to the vertices of two of the groups, each group is incident to two colors.

If the remainder is 1, then $k = \frac{n-1}{3}$. One color is incident to $2k$ vertices, while two other colors are incident to $2k + 1$ vertices each. If the ``smaller'' color subgraph with $2k$ vertices can accommodate one third of the edges, distributing the rest of the edges evenly to the two ``bigger'' color subgraphs is trivial. Otherwise, it is not possible to color the edges quite evenly, and the remaining over two thirds of edges still need to be shared between the bigger color groups. More precisely, the exact optimum can be written as
\begin{eqnarray}\nonumber
\textrm{OPT}(K_n) & = & \max \left( \left\lceil \frac{1}{3} |E(K_n)| \right\rceil, \left\lceil \frac{2 k(k+1) + \frac{k(k+1)}{2}}{2} \right\rceil \right) \\ \nonumber
& = & \max \left( \left\lceil \frac{1}{3} |E(K_n)| \right\rceil, \left\lceil \frac{5}{4}k(k+1) \right\rceil \right).
\end{eqnarray}

If the remainder is 2, then $k = \frac{n-2}{3}$. There are two colors incident to $2k + 1$ vertices and one color incident to $2k + 2$ vertices. Achieving the lower bound is possible, if the bigger color subgraph can avoid coloring more than one third of all edges. If not, the minimum size of the bigger color group is the optimum, that is
\begin{eqnarray}\nonumber
\textrm{OPT}(K_n) & = & \max \left( \left\lceil \frac{1}{3} |E(K_n)| \right\rceil, (k+1)^2 \right).
\end{eqnarray}

\subsection{Biclique}\label{sssect:biclique}

In this subsection we give a lower bound for min-max edge 2-coloring of a biclique $K_{m,n}$. Throughout this subsection we use $V_1$ and $V_2$ to denote the two independent sets of a biclique $K_{m,n}$. We also assume that $m \geq n$. The labels are chosen so that $|V_1| = m$ and $|V_2| = n$. Furthermore, we denote the set of vertices in $V_i$ incident to color $c$ by $V_i^c$.

As bicliques are not that different from cliques, the upcoming proofs are somewhat reminiscent of those in the previous subsection. For instance, the following lemma states that the absence of a global color implies at most four colors. We use this result in a similar fashion to how we used the three lemmas in the previous proof.

\begin{lemma}\label{lemma:biclique}
In an edge 2-coloring of a biclique $K_{m,n} = (V_1 + V_2, E)$, assume that no color is incident to all vertices. Then there can be at most four distinct colors.
\end{lemma}
\begin{proof}
Consider color $c_0$. By the assumption, either $V_1^{c_0} \subset V_1$ or $V_2^{c_0} \subset V_2$ (or both). For clarity, we assume $V_2^{c_0} \subset V_2$. Now, every vertex in $V_1^{c_0}$ has to be incident to some other colors, in order to color edges between $V_1^{c_0}$ and $V_2 \setminus V_2^{c_0}$. This makes each $v \in V_1^{c_0}$ saturated with two colors. $V_1^{c_0}$ can be incident to at most two different colors other than $c_0$, since all those other colors will be incident to all $v \in V_2 \setminus V_2^{c_0}$. If $V_1^{c_0} = V_1$, there can be no more colors. Otherwise, we can repeat the arguments so far, swapping 1s and 2s in place. This leads to the conclusion that all vertices in $V_2^{c_0}$ are also saturated with at most two other colors.

If $V_1^{c_0}$ (or $V_2^{c_0}$) is incident to exactly two other colors, the remainder $V_2 \setminus V_2^{c_0}$ ($V_1 \setminus V_1^{c_0}$) will be saturated by them. This makes all of $V_2$ ($V_1$) saturated, allowing no more colors. If also $V_2^{c_0}$ ($V_1^{c_0}$) is incident to exactly two other colors, at least one of them must be the same as one of the other colors incident to $V_1^{c_0}$ ($V_2^{c_0}$), so that the remainders can share a color. Thus, we have at most four colors.

If both $V_1^{c_0}$ and $V_2^{c_0}$ are incident to exactly one other color, those other colors must be different. Otherwise, the other color would necessarily be global by previous arguments, contradicting the assumption. So we have $V_1 \setminus V_1^{c_0}$ and $V_2 \setminus V_2^{c_0}$ incident to different colors. The only uncolored edges at this point are the ones between these two remainders. Having a new color incident to any vertex in the remainders leads to the opposing remainder to be saturated, allowing no more colors. Thus we end up with at most four colors, and all possibilities are now examined. \qed
\end{proof}

This lemma suffices as leverage for the proof of the following lower bound. The idea is again to show, that with a global color implied by five or more colors, it is impossible to get below the suggested lower bound.

\begin{theorem}\label{theorem:biclique}
For min-max edge 2-coloring, the following holds:
\begin{equation}\label{eqn:biclique-lb}
\textrm{OPT}(K_{m,n}) \geq \left\lceil \frac{1}{4} |E(K_{m,n})| \right\rceil = \left\lceil \frac{mn}{4} \right\rceil
\end{equation}
\end{theorem}
\begin{proof}
The only chance of having a smaller OPT than suggested is by having more than four colors in an optimal coloring. By Lemma \ref{lemma:biclique}, this is possible only if there is a global color. It is impossible to have more than one global color (unless there are only two colors), since two global colors already saturate every vertex.

Next, we show that if we restrict the coloring to have one global color, it is impossible to achieve even the lower bound in (\ref{eqn:biclique-lb}). Denote the global color by $F$, and the other colors by $c_i$, $i \in \{1, \ldots, C\}$,  where $C > 4$ is the number of non-global colors. No two non-global colors can be incident to common vertices, that is, $V_l^{c_i} \cap V_l^{c_j} = \emptyset, l \in \{1,2\}, i \neq j$. For convenience, we define $k_i^{c_j} := |V_i^{c_j}|$, and $\alpha_i^{c_j} \in [0,1]$ such that $\alpha_i^{c_j}k_i^F = k_i^{c_j}$. Note, that $k_1^F = m$ and $k_2^F = n$.

Our approach is to try to force the global color group as small as possible, while keeping the other color groups just small enough, that is,
\begin{equation}\label{eqn:biclique-non-globals}
k_1^{c_i}k_2^{c_i} = \alpha_1^{c_i}m\alpha_2^{c_i}n \leq \frac{mn}{4} \ \Longrightarrow \ \alpha_1^{c_i}\alpha_2^{c_i} \leq \frac{1}{4}.
\end{equation}
In order to make analysis simpler, we allow the values $k_i^{c_j} \geq 0$ to be fractional. Since this relaxation makes the set of feasible values for $k_i^{c_j}$ only bigger, the upcoming failure to even then get below the lower bound suffices for proof.

Consider a feasible coloring, where every non-global color group is smaller than the suggested lower bound, i.e. the inequalities in (\ref{eqn:biclique-non-globals}) are strict. In such a situation, we can always make the global color group smaller by following adjustments. We take the biggest $\alpha_1^{c_i}$, and grow $\alpha_2^{c_i}$ at the expence of other $\alpha_2^{c_j}, \ j \neq i$, until either $\alpha_2^{c_i} = 1$ or $\alpha_1^{c_i}\alpha_2^{c_i} = \frac{1}{4}$. We change sides and repeat, and end up with $\alpha_1^{c_i}\alpha_2^{c_i} = \frac{1}{4}$. This procedure makes the total size of the non-global color groups bigger (and thus the global color group smaller), since the color group of $c_i$ grows faster than the other groups shrink in the exchange.

The important point here is that a coloring that minimizes the size of the global color group, must have at least one non-global color group for which (\ref{eqn:biclique-non-globals}) holds with equality (given the relaxation of $k_i^{c_j}$). Let $c_1$ be such a color, and $\alpha_1^{c_1} \geq \alpha_2^{c_1}$. From (\ref{eqn:biclique-non-globals}) it follows, that $\alpha_1^{c_1} \geq \frac{1}{2}$ and $\alpha_2^{c_1} \leq \frac{1}{2}$. The edges between $V_1^{c_1}$ and $V_2 \setminus V_2^{c_1}$ must be colored with the global color. Thus, the size of the global color group is always at least
\begin{equation}
\alpha_1^{c_1}k_1^F(1 - \alpha_2^{c_1})k_2^F \geq \frac{1}{2}m \frac{1}{2}n = \frac{mn}{4}.
\end{equation}
In conclusion, having five or more colors in an edge $q$-coloring of a biclique implies exactly one global color, which in turn makes it impossible to achieve the lower bound in (\ref{eqn:biclique-lb}). Finally, falling back to the realm of integral solutions gives

\begin{equation}
\textrm{OPT}(K_{m,n}) \geq \left\lceil \frac{mn}{4} \right\rceil = \left\lceil \frac{1}{4}|E(K_{m,n})| \right\rceil.
\end{equation} \qed
\end{proof}

As in the case of cliques, the lower bound is often tight. For example, when $m$ and $n$ are both even, it is easy to find an optimal coloring, as illustrated in figure \ref{img:biclique-coloring}. The idea is to split the vertex sets $V_1$ and $V_2$ into equal sized halves. Then, the edges between each pair of halves on opposite sides can be colored with a distinct color. Even if $m$ and $n$ were odd, the aforementioned procedure leads to an asymptotically optimal coloring.

\begin{figure}[htb]
\centering \includegraphics[scale=0.75]{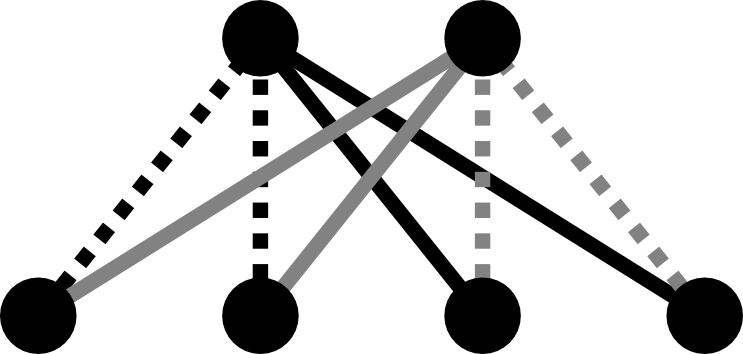}
\caption{An optimal coloring of $K_{2,4}$.}\label{img:biclique-coloring}
\end{figure}

\subsection{Hypercube}\label{sssect:hypercube}

In this subsection we give a lower bound for an optimal min-max edge 2-coloring of a hypercube $Q_n$. Also the tightness of the bound is discussed for both even and odd $n$. We begin by looking at subgraphs of $Q_n$ with $k$ vertices and the maximum number of edges they can have. Later we apply this result directly to color subgraphs with $k$ vertices.

\begin{lemma}\label{lemma:hypercube1}
In a hypercube $Q_n$, any subgraph with $k \leq |V(Q_n)|$ vertices has at most $\frac{1}{2}k \log_2 k$ edges. In other words, the average degree of such a subgraph is at most $\log_2 k$.
\end{lemma}
\begin{proof}
We prove the lemma by induction. We take the initial step by looking at the case $n=2$. A subgraph with $n-1=1$ vertex has $0 \leq \frac{1}{2} \log_2 1$ edges, so the lemma holds. The induction hypothesis is that the lemma holds for $Q_{n-1}$ and smaller hypercubes.

Now we take the induction step. Consider the hypercube $Q_n$. It can be partitioned into two subgraphs identical to $Q_{n-1}$, denote them by $Q_1$ and $Q_2$. Next, consider a subgraph $S$ of $Q_n$ with $k$ vertices. Denote the number of vertices of $S$ in $Q_1$ and $Q_2$ by $k_1$ and $k_2$, respectively. By the induction hypothesis, there are at most $\frac{1}{2}k_i \log_2 k_i$ edges among each of the two $k_i$ sized subgraphs of $S$. Additionally, there are at most $\min(k_1, k_2)$ edges between the two subgraphs of $S$, since each vertex in one of the hypercubes $Q_1$ and $Q_2$ is adjacent to exactly one vertex in the other. Due to symmetry and for simplicity, we can choose $k_1$ to be the smaller one. We also choose $\alpha$ so that $k_1 = \alpha k$. Consequently, $k_2 = (1 - \alpha)k$ and $\alpha \in [0,\frac{1}{2}] $. An upper bound for the number of edges in $S$ is thus as follows.
\begin{eqnarray}\label{eqn:hypercube-subsize}\nonumber
|E(S)| &\leq & \frac{1}{2}k_1 \log_2 k_1 + \frac{1}{2}k_2 \log_2 k_2 + k_1 \\ \nonumber
&& = \frac{1}{2} \log_2 (k_1^{k_1}k_2^{k_2}) + \frac{1}{2} \cdot 2\log_2 2^{k_1} \\ \nonumber
&& = \frac{1}{2} \log_2 \left((4k_1)^{k_1}k_2^{k_2}\right) \\ \nonumber
&& = \frac{1}{2} \log_2 \left((4\alpha k)^{\alpha k}\left((1-\alpha)k\right)^{(1-\alpha)k}\right) \\ \nonumber
&& = \frac{1}{2} k \log_2 \left((4\alpha)^{\alpha}(1-\alpha)^{(1-\alpha)} k^{\alpha}k^{(1-\alpha)}\right) \\ \nonumber
&& = \frac{1}{2} k \log_2 \left((4\alpha)^{\alpha}(1-\alpha)^{(1-\alpha)} k\right)
\end{eqnarray}
Now, if the right-hand side of (\ref{eqn:hypercube-subsize}) is shown to be less than or equal to $\frac{1}{2} k \log_2 k$, we are done.
\begin{eqnarray}\nonumber
& \displaystyle\frac{1}{2} k \log_2 \left((4\alpha)^{\alpha}(1-\alpha)^{(1-\alpha)} k\right) & \leq \ \frac{1}{2} k\log_2 k \\ \nonumber
\Longleftrightarrow & (4\alpha)^\alpha(1-\alpha)^{(1-\alpha)} & \leq \ 1 \\ \nonumber
\Longleftrightarrow & 4\alpha (1-\alpha)^{\frac{(1-\alpha)}{\alpha}} & \leq \ 1.
\end{eqnarray}
Observe that $4\alpha \leq 2$, since $\alpha \in [0,\frac{1}{2}]$. Thus it remains to show that $(1-\alpha)^{\frac{(1-\alpha)}{\alpha}} \leq \frac{1}{2}$. For this end, we make the following change of variables: $\beta = \frac{\alpha}{1-\alpha}$, which yields $1 - \alpha = \frac{1}{1 + \beta}$ and $\beta \in [0,1]$.
\begin{eqnarray}\nonumber
& (1-\alpha)^{\frac{(1-\alpha)}{\alpha}} & \leq \ \frac{1}{2} \\ \nonumber
\Longleftrightarrow & \left(\displaystyle\frac{1}{1 + \beta}\right)^{\frac{1}{\beta}} & \leq \ \frac{1}{2} \\ \nonumber
\Longleftrightarrow & 2^{\beta} & \leq \ 1 + \beta.
\end{eqnarray}
Since $2^{\beta}$ is convex, $1 + \beta$ is linear and equality holds, when $\beta = 0,1$, the above equation holds while $\beta \in [0,1]$. This concludes the proof. \qed
\end{proof}

The following lemma reveals, that the average degree of the whole graph bounds the maximum average degree of the color subgraphs from below. An intuitive reason for this is that otherwise there might not be enough edges in the color subgraphs to account for the edges of the original graph.

\begin{lemma}\label{lemma:hypercube2}
In a feasible edge $q$-coloring of $G$, there must be at least one color subgraph, whose average degree is greater or equal to $d_G/q$, where $d_G$ is the average degree of G.
\end{lemma}
\begin{proof}
Assume the opposite. Consider a feasible coloring with $m$ distinct colors and each color subgraph having smaller average degree than $\frac{d_G}{q}$. Let $n = |V(G)|$, $k_1, \ldots, k_m$ the number of vertices in each color subgraph and $d_1, \ldots, d_m$ the average degrees of the color subgraphs. Since the coloring is feasible, the number of edges in $G$ can be written as follows.
\begin{displaymath}
|E(G)| = \frac{nd_G}{2} = \sum_{i=1}^n \frac{k_id_i}{2} < \sum_{i=1}^m \frac{k_id_G}{2q} = \frac{d_G}{2q} \sum_{i=1}^m k_i \leq \frac{nd_G}{2} = |E(G)|.
\end{displaymath}
The second inequality follows from the fact that each vertex can be in at most $q$ different color subgraphs, so the sum over $k_i$ is at most $qn$. Having a contradiction, the lemma follows. \qed
\end{proof}

The two previous lemmas make the proof of the next theorem relatively straightforward.

\begin{theorem}\label{theorem:hypercube}
For min-max edge 2-coloring, the following holds:
\begin{equation}\label{eqn:hypercube-lb}
\textrm{OPT}(Q_{n}) \geq \frac{1}{2}n 2^{\frac{1}{2}n - 1} = \frac{1}{2} k \log_2 k,
\end{equation}
where $k = 2^{\frac{1}{2}n}$.
\end{theorem}
\begin{proof}
The right-hand side of (\ref{eqn:hypercube-lb}) is equal to the maximum number of edges in a subgraph of $Q_n$ with $k$ vertices, as follows from Lemma \ref{lemma:hypercube1}. The lemma also implies, that the average degree of any subgraph is smaller than $\log_2 k$, if it has less than $k$ vertices. A subgraph with $k$ or more vertices and less than $\frac{1}{2} k \log_2 k$ edges also has smaller average degree than $\log_2 k$. So, in a coloring where each color subgraph has less than $\frac{1}{2} k \log_2 k$ edges, the average degrees of the color subgraphs are all less than $\log_2 k = \frac{n}{2}$. Since $n$ is the average degree of $Q_n$, such an edge 2-coloring cannot be feasible according to Lemma \ref{lemma:hypercube2}. \qed
\end{proof}

\begin{figure}[htb]
\centering \includegraphics[scale=0.75]{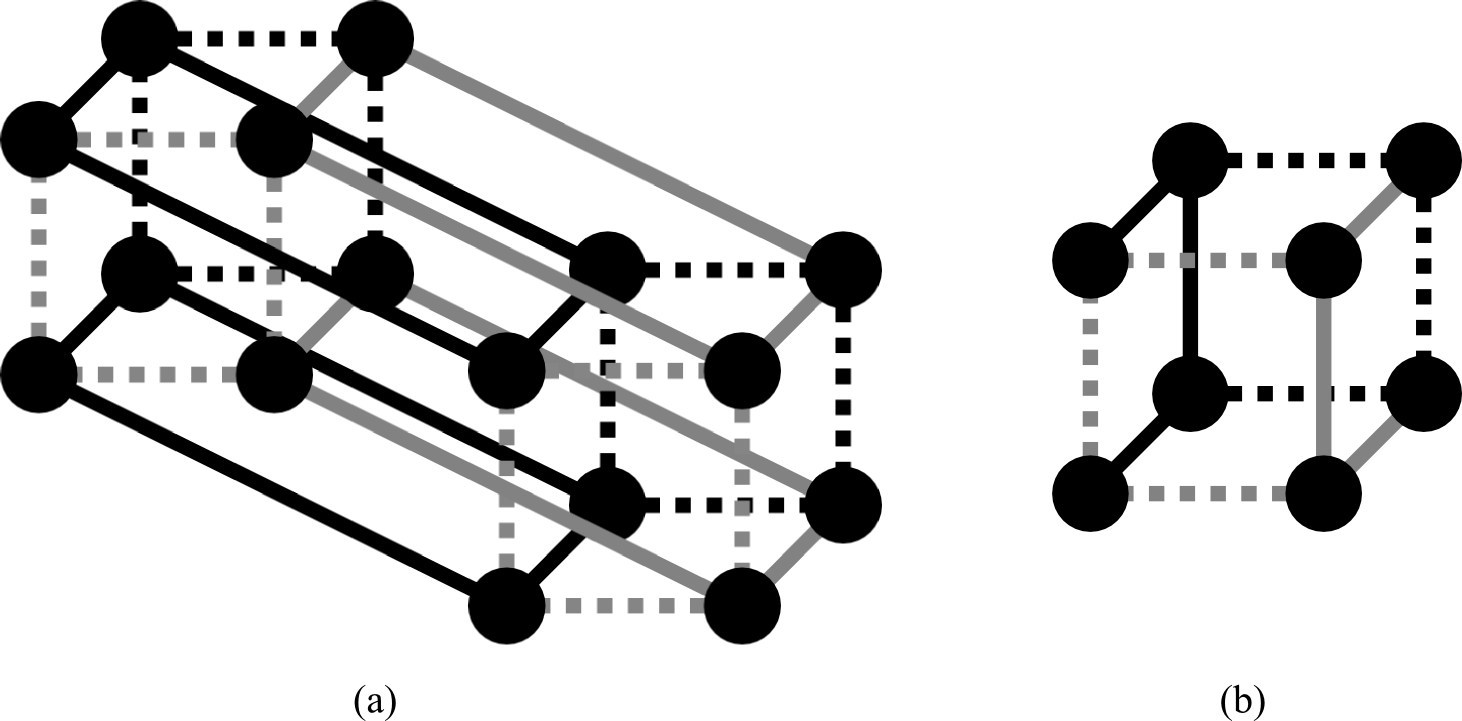}
\caption{An optimal coloring of (a) $Q_4$ and (b) $Q_3$. In (a) colors are reused to avoid complicated patterns.}\label{img:hypercube-coloring}
\end{figure}

The lower bound is tight for even $n$. A feasible coloring satisfying (\ref{eqn:hypercube-lb}) with equality for $n = 2m$ is constructed as follows, for example. Consider the bitstrings of length $n$ representing the vertices. Split the string into two halves of length $m$. Keeping, say, the left half fixed and cycling through the possible bit values for the right half gives a set of $2^m$ vertices, inducing a subgraph identical to an $m$-cube. Furthermore, going through all possible fixed strings on the left side gives $2^m$ disconnected $m$-cubes. Color each of these with a distinct color. At this point, every vertex is incident to exactly one color. Repeat the process, this time keeping the right side of the bitstring fixed. We get $2^m$ disconnected $m$-cubes consisting of the remaining uncolored edges. Again, color each cube with a distinct color, introducing exactly one new color to each vertex. Now all edges are colored, and each vertex is incident to exactly two colors, so the coloring is feasible. The size of each color group is $m2^{m-1}$, which satisfies (\ref{eqn:hypercube-lb}) with equality. 

For odd $n$, there is a coloring with color group size $(2m + 1)2^{m-1}$, where $m = \left\lfloor \frac{n}{2} \right\rfloor$. We achieve this as follows. First, we take two identically and optimally colored $(n-1)$-cubes. For each color, we have an $m$-cube of that color in both bigger cubes. For each such pair of $m$-cubes we add $2^{m-1}$ edges of the same color between corresponding vertices until we have an $n$-cube. Note that the size of each color group is now $2m2^{m-1} + 2^{m-1} = (2m + 1)2^{m-1}$. Whether this is an optimal coloring, is an open question, although the existence of a better coloring seems unlikely. Example colorings of $n$-cubes for both even and odd $n$ are presented in figure \ref{img:hypercube-coloring}.

\section{Lower and Upper Bounds}\label{sect:lb-ub}

In this section we present two lower bounds for the optimum (OPT) of the min-max edge $q$-coloring problem, and we show that a trivial coloring algorithm achieves a linear approximation factor in the number of vertices in the graph.


We begin with a lower bound in terms of maximum degree. The bound is very simple, but nevertheless useful in some proofs.
\begin{theorem}\label{theorem:max-deg-lb}
Denote the maximum degree of graph $G$ by $\Delta(G)$. Then,
\begin{equation}
\textrm{OPT} \geq \left\lceil \frac{\Delta(G)}{q} \right\rceil.
\end{equation}
\end{theorem}
\begin{proof}
The theorem follows directly from the fact that only $q$ different colors can be incident to any vertex of $G$. \qed
\end{proof}


The next lower bound is in terms of average degree. This bound is rather loose for graphs with a small average degree, but becomes tighter as the graphs get denser.
\begin{theorem}\label{theorem:avg-deg-lb}
Let the average degree of $G$ be denoted by $d(G)$. Then,
\begin{displaymath}\label{eqn:avg-deg-lb}
\textrm{OPT} \geq \frac{d^2(G)}{2q^2}.
\end{displaymath}
\end{theorem}

\begin{proof}
For convenience, we denote OPT by $m$. The idea is to find an upper bound for the average degree of $G$ in terms of $m$, which in turn yields a lower bound for $m$ in terms of the average degree.

First we show that the average degree of a graph with at most $k$ edges is at most $\sqrt{2k}$. If $k=1$ and there are $n$ vertices, the average degree is certainly less than $\sqrt{2k}$. Observe that an $n$-clique $K_n$ has the largest possible average degree, given at most $n$ vertices or $|E(K_n)|$ edges. If $k = |E(K_n)|$, we get
\begin{displaymath}
k = \frac{n(n-1)}{2} \quad \Rightarrow \quad n = \frac{1}{2} + \sqrt{\frac{1}{4} + 2k} \quad \Rightarrow \quad d(K_n) = n - 1 \leq \sqrt{2k}.
\end{displaymath}
If we keep $n$ fixed and add edges (until we have a clique), then the average degree grows linearly in $k$. Since $\sqrt{2k}$ is convex, it is larger than the average degree of any graph with $k$ edges. 


Since each subgraph induced by a color group in an optimal coloring has at most $m$ edges, their average degrees are at most $\sqrt{2m}$. Furthermore, $d(G)$ is maximized, if each vertex is in $q$ color subgraphs, whose average degree is $\sqrt{2m}$. Let $n = |V(G)|$. We get
\begin{displaymath}
d(G) = \frac{2|E(G)|}{|V(G)|} \leq \frac{2qn\sqrt{2m}}{2n} = q\sqrt{2m}.
\end{displaymath}
Thus, there is no graph with higher average degree than $q\sqrt{2\textrm{OPT}}$. The claim follows:
\begin{displaymath}
\textrm{OPT} \geq \frac{d^2(G)}{2q^2}.
\end{displaymath} \qed
\end{proof}


In the following we show that the approximation factor of the most trivial algorithm for min-max edge $q$-coloring is linear in the number of vertices. Here is the definition of the trivial algorithm, followed by the theorem stating the approximation factor.

\begin{algorithm}
\caption{Trivial coloring algorithm}
\label{alg:trivial}
\textbf{Input:} Graph $G$ \\
1. Assign the same color to each edge of $G$ \\
2. $m \longleftarrow |E(G)|$ \\
\textbf{Output:} $m$
\end{algorithm}

\begin{theorem}\label{theorem:factor-ub}
The approximation factor of Algorithm \ref{alg:trivial} is $O(n)$, where $n$ is the number of vertices in the input graph.
\end{theorem}
\begin{proof}
Algorithm \ref{alg:trivial} achieves objective function value $m = |E(G)| = \frac{1}{2}nd(G)$, where $d(G)$ is the average degree of $G$. By Theorem \ref{theorem:avg-deg-lb} and by making the restriction $d(G) \geq n^\alpha$, we get
\begin{displaymath}
\frac{m}{\textrm{OPT}} \leq \frac{8nd(G)}{2d^2(G)} = \frac{4n}{d(G)} \leq 4n^{(1-\alpha)}.
\end{displaymath}
Choosing $\alpha \geq 0$ yields $d(G) \geq 1$, which is the case for any connected graph with $n \geq 2$ (every vertex has at least one edge incident to it). Thus, the approximation factor is at most $4n$, that is, $O(n)$. \qed
\end{proof}

\section{Approximation Algorithm for Planar Graphs}\label{sect:separator}


In this section we present an algorithm that achieves a sublinear approximation factor for planar graphs. The basic idea of the algorithm comes from the following theorem, proven by Lipton and Tarjan \cite{Lipton1980}.

\begin{theorem}\label{theorem:separator}
Let $G$ be an $n$-vertex planar graph and let $0 \leq \epsilon \leq 1$. Then there is some set $S$ of $O(\sqrt{n / \epsilon})$ vertices whose removal leaves $G$ with no connected component with more than $\epsilon n$ vertices. Furthermore the set $S$ can be found in polynomial time.
\end{theorem}

The so called planar separator $S$ described in the above theorem is particularly useful regarding min-max edge $q$-coloring, since the residue components are balanced. Moreover, $S$ itself does not have too many vertices either. Now we proceed to define the algorithm.

\begin{algorithm}[htb]
\caption{Planar separator 2-coloring algorithm}
\label{alg:separator}
\textbf{Input:} A planar graph $G$ with $n$ vertices
\begin{tabbing}
1. Find a separator $S$ as described in Theorem \ref{theorem:separator}, with $\epsilon = n^{-1/3}$ \\
2. Color edges incident to $S$ with one color \\
3. $m \longleftarrow$ number of edges incident to $S$ \\
3. Remove $S$ from $G$ \\
4. \= For each remaining connected component $S_i$: \\
   \> 5. Color edges incident to $S_i$ with a unique color \\
   \> 6. $m_i \longleftarrow$ number of edges incident to $S_i$ \\
   \> 7. If $m_i > m$, $m \longleftarrow m_i$
\end{tabbing}
\textbf{Output:} $m$
\end{algorithm}

Choosing the order of magnitude of $\epsilon$ is central in obtaining the lowest possible approximation factor, as will become clear in the proof of the following theorem.

\begin{theorem}
The approximation factor of Algorithm \ref{alg:separator} is $O(n^{2/3})$.
\end{theorem}
\begin{proof}
Denote the maximum degree of the input graph $G$ by $\Delta$. After the algorithm ends, there are two possibilities. Either one of the colors associated with the separated components or the color associated with the separator has the most vertices. Furthermore, each vertex can be incident to at most $\Delta$ edges. Thus,
\begin{displaymath}
m \leq \max \left( \Delta O(\sqrt{n/\epsilon}), \Delta \epsilon n \right).
\end{displaymath}
Theorem \ref{theorem:max-deg-lb} gives us a lower bound for OPT in terms of $\Delta$. Together with the above, the approximation factor is
\begin{displaymath}
\frac{m}{\textrm{OPT}} \leq \max \left( O(\sqrt{n/\epsilon}), 2\epsilon n \right).
\end{displaymath}
The order of magnitude of the right-hand-side is minimized when it is equal for both terms. This happens when $\epsilon$ is chosen to be $n^{-1/3}$, as is done on the first line of the algorithm. We get
\begin{displaymath}
\frac{m}{\textrm{OPT}} \leq \max \left( O(\sqrt{n^{4/3}}), 2n^{2/3} \right) = O(n^{2/3}).
\end{displaymath}
\qed
\end{proof}

\section{Summary}\label{sect:summary}

The goal of this paper is to analyze the problem of efficiently allocating channels in wireless mesh networks from a theoretic point of view and to design and analyze some basic approximation algorithms. The analysis is simplified by modelling the channel allocation problem as a graph coloring problem, namely min-max edge $q$-coloring. The concept of edge $q$-coloring captures the restriction in some proposed WMN architectures, where each network node can use at most a number of different frequency channels at once. Furthermore, we give the most attention to the case $q=2$, since it has been considered important from a practical perspective.

For the min-max edge $q$-coloring problem, we prove NP-hardness, both in a more general case (see Problem \ref{problem:general-mme}), where each vertex has its individual value for $q$, and in the case where the value of $q \geq 2$ is constant for each vertex.  We show lower bounds for the optimum in terms of maximum and average degree. We also introduce two new algorithms: an approximation algorithm for planar graphs, and an exact polynomial time algorithm for trees. The former is shown to have an approximation factor of $O(n^{2/3})$. 
We also give lower bounds that are close, and often tight, to optimums of three special cases, namely cliques, bicliques and hypercubes.


Interesting directions for future research include finding hardness of approximation results and better algorithms, especially for min-max edge $q$-coloring on general graphs. Also it might be interesting to see how the proposed algorithms would affect performance, if applied to actual Wireless Mesh Networks.

\bibliographystyle{plain}
\bibliography{icalp2013}


\end{document}